\documentclass[letterpaper,onecolumn,allowtoday,unpublished]{quantumarticle}
\pdfoutput=1
\usepackage[utf8]{inputenc}
\usepackage[english]{babel}
\usepackage[T1]{fontenc}
\usepackage[numbers,sort&compress]{natbib}
\usepackage{braket}
\usepackage{amsmath}
\usepackage{amssymb} 
\usepackage{tikz}
\usetikzlibrary{quantikz}
\usepackage{complexity}
\usepackage{comment}

\usepackage{algorithm}
\usepackage[noend]{algpseudocode}
\algrenewcommand\algorithmicrequire{\textbf{Input:}}
\algrenewcommand\algorithmicensure{\textbf{Output:}}
\usepackage{tabularx}
\makeatletter
\newcommand{\multiline}[1]{%
  \begin{tabularx}{\dimexpr\linewidth-\ALG@thistlm}[t]{@{}X@{}}
    #1
  \end{tabularx}
}
\makeatother

\usepackage{amsthm}

\newtheorem{lemma}{Lemma}
\newtheorem{prop}{Proposition}

\usepackage[export]{adjustbox}
\usepackage{subfig}

\usepackage{hyperref}
\usepackage[capitalise]{cleveref}

\DeclareMathOperator*{\ex}{\mathbb{E}}

\DeclareMathOperator{\sgn}{\mathrm{sgn}}

\NewDocumentCommand{\trace}{m o}{%
	\ \!\mathrm{tr}#1 %
	\IfValueT{#2}{\ensuremath{_{#2}}}%
}

\newcommand{\tp}{\mathcal{B}_{\mathrm{TP}}}
\newcommand{\nontp}{\mathcal{B}_{\mathrm{NTP}}}

\newcommand{\mK}{\mathcal{K}}
\newcommand{\mX}{\mathcal{X}}
\newcommand{\mS}{\mathcal{S}}
\newcommand{\mH}{\mathcal{H}}
\newcommand{\mU}{\mathcal{U}}
\newcommand{\mV}{\mathcal{V}}

\usepackage{xcolor}
\definecolor{ao}{RGB}{253,141,60}

\newcommand{\cB}{\mathcal{B}}
\newcommand{\cT}{\mathcal{T}}

\newcommand{\tr}{\mathrm{Tr}}
\newcommand{\da}[1]{\left\|#1\right\|_\Diamond}
\newcommand{\trn}[1]{\left\|#1\right\|_{tr}}
\newcommand{\iu}{\mathrm{i}}
\begin{document}

\title{%
\mbox{Quasiprobabilistic imaginary-time evolution on} \mbox{quantum computers}
}

\author{Annie Ray}
\affiliation{Institute for Quantum Computing, University of Waterloo, Waterloo ON, N2L 0A4, Canada}
\affiliation{Department of Physics and Astronomy, University of Waterloo, Canada}
\affiliation{Perimeter Institute for Theoretical Physics, 31 Caroline St N, Waterloo ON, N2L 2Y5, Canada}
\author{Esha Swaroop}
\affiliation{Institute for Quantum Computing, University of Waterloo, Waterloo ON, N2L 0A4, Canada}
\affiliation{Department of Physics and Astronomy, University of Waterloo, Canada}
\affiliation{Perimeter Institute for Theoretical Physics, 31 Caroline St N, Waterloo ON, N2L 2Y5, Canada}
\author{Ningping Cao}
\affiliation{Perimeter Institute for Theoretical Physics, 31 Caroline St N, Waterloo ON, N2L 2Y5, Canada}
\affiliation{Institute for Quantum Computing, University of Waterloo, Waterloo ON, N2L 0A4, Canada}
\email{ncao@uwaterloo.ca}
\author{Michael Vasmer}
\affiliation{Inria Paris, 48 rue Barrault, Paris 75013, France}
\affiliation{Perimeter Institute for Theoretical Physics, 31 Caroline St N, Waterloo ON, N2L 2Y5, Canada}
\affiliation{Institute for Quantum Computing, University of Waterloo, Waterloo ON, N2L 0A4, Canada}
\author{Anirban Chowdhury}
 \affiliation{IBM Quantum, IBM T.\ J.\ Watson Research Center, Yorktown Heights, NY, 10598, USA}
 \email{anirban.chowdhury@ibm.com}

\date{\today}

\maketitle

\begin{abstract}

Imaginary-time evolution plays an important role in algorithms for computing ground-state and thermal equilibrium properties of quantum systems, but can be challenging to simulate on classical computers. Many quantum algorithms for imaginary-time evolution have resource requirements that are prohibitive for current quantum devices and face performance issues due to noise. Here, we propose a new algorithm for computing imaginary-time evolved expectation values on quantum computers, inspired by probabilistic error cancellation, an error-mitigation technique. Our algorithm works by decomposing a Trotterization of imaginary-time evolution into a probabilistic linear combination of operations, each of which is then implemented on a quantum computer. The measurement data is then classically post-processed to obtain the expectation value of the imaginary-time evolved state. Our algorithm requires no ancillary qubits and can be made noise-resilient without additional error-mitigation. It is well-suited for estimating thermal expectation values by making use of the notion of a thermal pure quantum state. We demonstrate our algorithm by performing numerical simulations of thermal pure quantum state preparation for the 1D Heisenberg Hamiltonian on 8 qubits, and by using an IBM quantum computer to estimate the energy of the same Hamiltonian on 2 qubits. We observe promising results compared to the exact values, illustrating the potential of our algorithm for probing the physics of quantum many-body systems on current hardware.

\end{abstract}

\section{Introduction}
Quantum computers are expected to have wide-ranging applications in various domains~\cite{daley2022Practical,bauer2020Quantum,orus2019Quantum,dalzell2023quantum} but perhaps most significantly in the simulation of quantum physical systems. 
Important problems in this arena concern the estimation of ground-state and thermal equilibrium properties of quantum many-body systems, problems which have applications in chemistry and material science~\cite{cao2019quantumchemistry,bauer2020Quantum,mcardle2020quantumcomputational,Babbush2018electronic}.
Consequently there has been an enormous amount of research \cite{Lemieux2021,Rall2020blockencodings,tubman2018postponing,Kivlichan2020} directed towards the development of quantum algorithms for these problems. Although computational complexity results limit the possibility for dramatic quantum speed-ups for these problems \cite{kitaev2002classical,gharibian2015survey,bravyi2022quantum}, there is also theoretical evidence that quantum computers potentially provide an advantage over classical algorithms \cite{gharibian2023dequantizing} .

The idea of \emph{imaginary-time evolution} plays a key role in many algorithms, both classical and quantum, for simulating ground-state and equilibrium properties~\cite{verstraete2004matrixproduct,zwolak2004mixedstate,wolf2015imaginarytime,lehtovaara2007solution,kraus2010generalized,mcclean2015compact,shi2018variational,mcclean2013feynmans,mcclean2015clock}.
Given a Hamiltonian $H$ and an inverse temperature $\beta \ge 0$, imaginary-time evolution (ITE) under $H$ is described by the exponential operator $e^{-\beta H}$. Performing ITE with sufficiently high $\beta$ will drive a quantum system to its ground state, making it a useful primitive in algorithms for ground state preparation. 
ITE is also used as a subroutine in quantum algorithms for preparing thermal equilibrium or Gibbs states \cite{poulin09sampling}, which find application in solving optimization problems. 
Performing ITE of quantum Hamiltonians using classical methods can be computationally prohibitive due to the exponential increase in Hilbert space dimension and the sign problem~\cite{signproblem1990manyelectron,Troyer2005sign}. Many existing quantum algorithms such as those based on block-encoding \cite{chowdhury2017quantum,mcardle2019variational,vanApeldoorn2020quantumsdpsolvers,holmes2022quantum,silva2023fragmented} require fault-tolerant quantum computers and are beyond the reach of current hardware~\cite{campbell2017roads}. 
Several recent works have therefore proposed quantum algorithms that are potentially more suitable for near-term to intermediate quantum devices~\cite{Motta2019DeterminingEA,sewell2210thermal,benedetti2021hardware,silva2024partition,hejazi2023Adiabatic}. 

In this work, we add to the quantum toolkit for imaginary-time evolution and develop a new near-term friendly algorithm using ideas from quantum error mitigation (QEM), a collection of classical post-processing techniques which aim to mitigate the effects of noise without requiring the full machinery of quantum error correction and fault-tolerance~\cite{temme2017Error,caiQuantumErrorMitigation2022}.
These techniques are usually not scalable, but have shown to significantly improve the performance of current devices~\cite{caiQuantumErrorMitigation2022}. In particular, our work is inspired by the error-mitigation technique of probabilistic error cancellation (PEC)~\cite{temme2017Error,endo2018Practical,kandala2019Error,piveteau2022Quasiprobability}. The original motivation of PEC as a technique is to invert a noise channel on noisy quantum devices \cite{temme2011quantum}. Probabilistic error cancellation simulates the action of an ideal, unitary operation from a linear combination of native channels implementable on a quantum device aided by post-processing. The basic idea underlying PEC is as follows.
Given a quantum device with a set of known native (and possibly noisy) gates, one decomposes a target unitary $U$ into a weighted sum of the native gates. The coefficients in the linear combination will be real numbers that sum to 1, although they may or may not be positive and hence form a \emph{quasi-probability distribution} (QPD). Nevertheless, the action of $U$ can be effected by a protocol that probabilistically samples a native gate from the linear combination. With sufficient number of samples and simple classical post-processing, the expectation value of an observable after the action of the target unitary $U$ can be reconstructed.  Recovering the ideal expectation value comes at an increased cost related to the negativity of the quasi-probability distribution. 
The success of the PEC method relies on an accurate characterization of the native gates in the quantum device.
This can be a very challenging task in practice, although recent work has shown that tailoring the device noise to Pauli noise (via randomized compiling~\cite{wallman2016Noise}) allows efficient and accurate noise characterization~\cite{vandenberg2023Probabilistic}.
We note that related QPD ideas have also been used in other quantum information contexts, such as classical simulation of quantum circuits~\cite{pashayan2015Estimating,howard2017Application,seddon2019Quantifying,heinrich2019Robustness,seddon2021Quantifying}, circuit cutting~\cite{mitarai2021Constructing,wiersema2022circuit,lowe2023Fast,brenner2023Optimal,piveteau2023Circuit}, and for mitigating noise in non-Clifford gates~\cite{piveteau2021Error}.

Here, we design a quasi-probabilistic algorithm inspired by PEC to perform imaginary-time evolution for a given Hamiltonian $H$. 
The starting point of our algorithm is the observation that we can implement non-unitary and even non-trace-preserving operations as a linear combination of native gates and measurement channels. 
We use our algorithm to estimate the expectation value of a given observable for imaginary-time evolved states, and we analyze its cost, specifically for the case of Trotterized imaginary time evolution. In addition, we show that our algorithm is suitable for preparing \textit{thermal pure quantum} (TPQ) states, pure states which can be used to estimate the thermal expectation value of observables; see \cref{fig:tpq-demo}.
We demonstrate the efficacy of our approach through classical simulations and demonstrations on quantum hardware.

\begin{figure}[ht]
    \centering
    \includegraphics[]{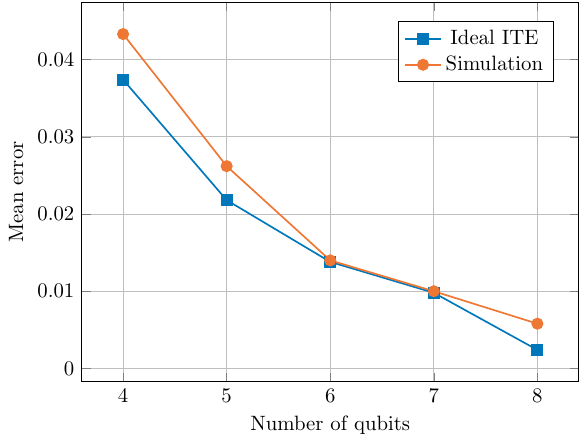}%
    \caption{Thermal expectation value estimation using ITE. For each data point we average the error in the thermal expecation value $|\langle \psi_i | O_j | \psi_i \rangle - \trace(|\psi_\beta\rangle\langle\psi_\beta|O_j)|$ over 30 randomly chosen Pauli observables $O_j$ and 10 random TPQ states $|\psi_i\rangle=e^{-0.02 H}U_i|0\rangle^{\otimes n}$, where $H$ is the 1D Heisenberg Hamiltonian on $n$ qubits and $U_i$ is a randomly-chosen Clifford unitary. The blue and orange points show the values for exact and simulated ITE, respectively, where the simulation implements \cref{alg:est-rescaled-exp}.}
    \label{fig:tpq-demo}
\end{figure}

Our method of simulating ITE in expectation using quasiprobability decompositions has two merits that make it amenable for implementation on near-term devices where qubits are few, and multi-qubit entangling gates are noisier than single-qubit gates.
First, our method does not need auxiliary qubits beyond those of the system described by a given Hamiltonian.
Second, we avoid the problem of implementing multi-qubit-controlled operations between auxiliary and data qubits~\cite{poulin2009thermal,poulin09sampling}, which can require large SWAP networks for implementation on devices that only permit entangling operations between nearest neighbors (such as present day superconducting devices).

The remainder of the paper is organized as follows.
In \cref{Sec:QPDGeneral} we review the quasiprobability decomposition of linear maps and discuss our extension to non-linear maps, including our algorithm for estimating rescaled expectation values.
Following this, in \cref{Sec:QPDforQITE} we use our algorithm for the specific case of imaginary time evolution and combine the protocol with various choices of initial states, focusing on the thermal pure quantum state. 
In this section we also present the results of our demonstrations conducted using IBM's quantum computing hardware.
Finally we provide a summary and discuss future work in \cref{sec:conc}.

\section{Implementing linear maps via quasiprobability decompositions}\label{Sec:QPDGeneral}
In particular, we consider the problem of implementing Hermiticity-preserving linear maps---which may not be trace-preserving in general---on a quantum device with a (basis) set of completely positive (CP) native operations.
In \cref{Sec:QPDGenAlgos} we present our algorithm for estimating rescaled expectation values of observables with respect to quantum states evolving under the linear map of interest.
In \cref{Sec:QPDGenAlgoAnalysis} we analyze the sampling cost of our algorithm.
Our method requires as input a decomposition of the linear map of interest into a linear combination of the set of available operations.
In \cref{Sec:QPDGenDecomp}, we discuss the conditions under which a linear map can be decomposed as such and we briefly review the method to obtain the linear decomposition.
And in \cref{Sec:QPDProd} we generalize our algorithm to the case of implementing a sequence of linear maps.

\subsection{Procedure for estimating rescaled expectation values of observables}\label{Sec:QPDGenAlgos}
Suppose we have a linear map $\cT$ over the $2^n$-dimensional Hilbert space, and a set $\{\cB_i\}$ of native operations on a quantum device such that for some real coefficients $q_i$,
\begin{align}\label{eq:qpd-gen-map}
    \cT(\rho) = \sum_i q_i \cB_i (\rho)
\end{align}
for all input states $\rho$.
Using \cref{eq:qpd-gen-map}, we can define parameters $\gamma$ and $p_i$ such that
\begin{align}\label{eq:qpd-params}
    \gamma &\coloneqq \sum_i |q_i|, \\
    p_i    &\coloneqq \frac{|q_i|}{\gamma}.
\end{align}
We refer to such a linear combination as a quasiprobability decomposition (QPD).
Note that since the $q_i$ are real, $\gamma \geq 0$ and the $p_i$ form a probability distribution.
Moreover, \cref{eq:qpd-gen-map} can be rewritten as 
\begin{align}\label{eq:qpd-rewrite}
    \cT(\rho) = \gamma \sum_i \sgn(q_i) \cdot p_i \cdot \cB_i (\rho).
\end{align}
This immediately suggests a way to probabilistically build the map $\cT$ by sampling from the set $\{\cB_i\}$ according to the probability distribution $p_i$ (see \cref{fig:qpdalgo}).
\begin{figure}[ht!]
    \centering
    \includegraphics{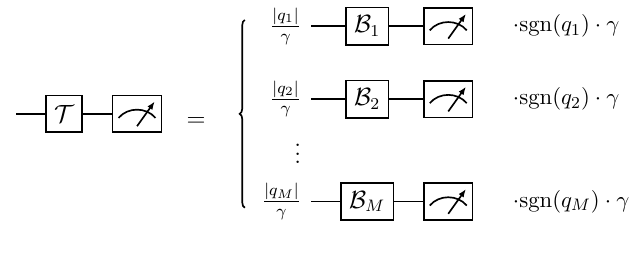}
    \caption{A schematic of a linear map $\cT$ implemented by sampling from its quasiprobability distribution (see \cref{eq:qpd-rewrite}, \cref{alg:est-rescaled-exp}). 
    Given a quasiprobability decomposition of the map $\cT$, basis circuits are randomly sampled according to the distribution $|q_i|/\gamma$, and their measurement outcomes are weighted by the factors $\mathrm{sgn}(q_i)\gamma$.
    }
    \label{fig:qpdalgo}
\end{figure}

In particular, we consider the task of estimating the expectation value of some observable $A$ (having operator norm $\|A\| \leq 1$) with respect to the state $\cT(\rho)$.
Since $\cT$ may be a non-trace-preserving map, the state $\cT(\rho)$ is unnormalized in general.
Thus, the rescaled expectation value of the observable $A$ is
\begin{align}\label{eq:rescaled-exp-vals}
    \langle A \rangle_{\cT(\rho)} := \frac{\trace{\left[A\cT(\rho)\right]}}{\trace{\left[\cT(\rho)\right]}}.
\end{align}
In order to correctly estimate the rescaled expectation value of $A$, we require estimates of the numerator and denominator of \cref{eq:rescaled-exp-vals} which are obtained using \cref{eq:qpd-rewrite} 
\begin{align}\label{eq:qpd-obs-gen-map}
    \trace{\left[A \cT(\rho) \right]} =&~ \gamma \sum_i \sgn(q_i) \cdot p_i \cdot \trace{\left[A \cB_i(\rho)\right]}, \\
    \trace{\left[\cT(\rho) \right]} =&~ \gamma \sum_i \sgn(q_i) \cdot p_i \cdot \trace{\left[\cB_i(\rho)\right]}.
\end{align}
Denoting the trace-preserving and non-trace-preserving operations in $\{\cB_i\}$ as $\tp$ and $\nontp$ respectively, the denominator in \cref{eq:rescaled-exp-vals} can be rewritten as 
\begin{align}\label{eq:rescaling-norm}
    \trace{\left[\cT(\rho) \right]} &= \gamma \left(\sum_{\cB_i \in \tp} \sgn(q_i) \cdot p_i +\sum_{\cB_i \in \nontp} \sgn(q_i) \cdot p_i \cdot \trace{\left[\cB_i(\rho)\right]} \right).
\end{align}
In this work, we consider $\nontp$ to comprise rank-1 projectors $\Pi_{\ket{\psi}} = \ket{\psi}\bra{\psi}$ on specific quantum states $\ket{\psi}$ (see \cref{table:EBLBasis} for example).
Typically such non-trace-preserving operations would be implemented via a POVM given by $\{\ket{\psi}\bra{\psi}, \ket{\psi^{\perp}}\bra{\psi^{\perp}}\}$ and post-selecting on the outcome $\ket{\psi}$.
Thus, $\trace{\left[\cB_i(\rho)\right]}$ is the probability of post-selecting the outcome corresponding to the projector for $\cB_i$.

If we randomly sample $N$ circuits from the probability distribution $p_i$ in \cref{eq:qpd-params}, the random variable $W_j := \gamma \sgn(q_i) I_j$ corresponding to the $j^{\mathrm{th}}$ circuit is an estimator (see \cref{eq:trace-estimator}) for $\trace{[\cT(\rho)]}$, where $j \in \{1, 2,...,N\}$.
Here, the indicator variable $I_j$ is 0 when the $j^{\mathrm{th}}$ circuit implements an operation $\cB_i \in \nontp$ whose associated POVM outcome is discarded; otherwise $I_j = 1$.

\begin{align}\label{eq:trace-estimator}
    \ex{[W_j]} = \ex{[\gamma \sgn(q_i)I_j]} &=  \gamma \sum_{i,I_j} \sgn(q_i) I_j \cdot \frac{|q_i|}{\gamma} \cdot \Pr[I_j] \\
    &= \gamma \sum_{i} \sgn(q_i) \cdot p_i \cdot \trace{[\cB_i{(\rho)}]}\\
    &= \trace{[\cT(\rho)]}
\end{align}
Similarly, we show in~\cref{app:algcosts} that the random variable $M_j = \gamma \sgn(q_i) A_j$, where $A_j$ is the outcome of measuring the observable $A$ for the $j^{\mathrm{th}}$ circuit, is an estimator of $\trace{[A\cT(\rho)]}$.
\begin{align}\label{eq:obs-estimator}
    \ex{[M_j]} = \trace{[A\cT(\rho)]}
\end{align}
Together, the two estimators $W_j, M_j$ allow us to estimate the rescaled expectation value in \cref{eq:rescaled-exp-vals}.
We formalize the procedure described above in \cref{alg:est-rescaled-exp}.

\begin{algorithm}[h!]
\caption{Estimate re-scaled expectations}\label{alg:est-rescaled-exp}
\begin{algorithmic}[1]
    \Require 
    \Statex QPD of $\cT$ as $\cT(\cdot) = \gamma \sum_i \sgn(q_i) p_i \cB_i (\cdot)$
    \Statex Initial state $\rho$
    \Statex Number of random samples $N$
    \Statex Observable $A$
    \Ensure Estimate of the re-scaled expectation value $\langle A \rangle_\cT$
        \For{$j=1$ to $N$}
            \State Sample $i$ from the distribution $ p_i $
            \If{$\cB_i$ is TP} \State apply $\cB_i$ directly to $\rho$
            \State $I_j \gets 1$ \Comment{Indicator variable}
            \Else
            \State Measure in the basis corresponding to $\cB_i$ 
            \State Store the outcome as $I_j \in \{0, 1\}$
            \EndIf
            \State $W_j \gets \gamma \cdot \sgn(q_i) \cdot I_j$ \Comment{Weighted outcome}
            \State Measure $A$
            \State Store the outcome as $A_j$
            \State $M_j \gets \gamma \cdot \sgn(q_i) \cdot A_j$
        \EndFor
        \State $\widetilde{\trace{\left[\cT(\rho)\right]}} \gets \sum_j W_j/N$ \Comment{Empirical estimate of $\trace{\left[\cT(\rho)\right]}$}
        \State $\widetilde{\trace{\left[A\cT(\rho)\right]}} \gets \sum_j M_j/N$ \Comment{Empirical estimate of $\trace{\left[A\cT(\rho)\right]}$}
        \State \Return $\widetilde{\trace{\left[A\cT(\rho)\right]}}/\widetilde{\trace{\left[\cT(\rho)\right]}}$
    \end{algorithmic}
\end{algorithm}

\subsection{Sampling cost of the algorithm}\label{Sec:QPDGenAlgoAnalysis}
The crux of \cref{alg:est-rescaled-exp} lies in the estimation of $\trace{[A\cT(\rho)]}$ and  $\trace{[\cT(\rho)]}$ separately, in order to estimate $\langle A \rangle_{\cT(\rho)}$ via their ratio as per \cref{eq:rescaled-exp-vals}. 
We use standard concentration inequalities to bound the number of samples $N$ needed to estimate $\trace{[A\cT(\rho)]}$ and  $\trace{[\cT(\rho)]}$ each within some precision $\epsilon$, which allows us to bound the error in the estimate of $\langle A \rangle_{\cT(\rho)}$ as per \cref{lem:qpd-cost}.

\begin{lemma}[Algorithm 1 cost]\label{lem:qpd-cost}
    Given a decomposition of a linear map $\cT$ into a set of operations as in \cref{eq:qpd-gen-map} and a Hermitian observable $A$ with $\|A\| \le 1$, \cref{alg:est-rescaled-exp} uses $N = \mathcal{O}\left(\frac{\gamma^2}{\epsilon^2}\right)$ random circuits drawn from the distribution $p_i$ in \cref{eq:qpd-params}, to return an estimate of the re-scaled expectation value $ \langle A \rangle_\cT$ within error
    \begin{align}
        \frac{2\epsilon}{|\trace{[\cT(\rho)]}|}\left(1 +  \frac{\epsilon}{|\trace{\left[\cT(\rho)\right]}|}\right).
    \end{align}
\end{lemma}
\begin{proof}[\bf Proof sketch of \cref{lem:qpd-cost}]
We note that for an observable $A$ with operator norm $||A|| \leq 1$, the random variables $W_j, M_j$ are such that
\begin{align}
    -\gamma \leq W_j, M_j \leq \gamma.
\end{align} 
For $N$ circuits (enumerated by the index $j$) sampled independently from $p_i$, we can define the following:
\begin{align}
    \overline{W} &= \frac{\sum_{j}W_j}{N}, \label{eq:numerator-bound} \\
    \overline{M} &= \frac{\sum_{j}M_j}{N} \label{eq:denom-bound} 
\end{align}
For some $\epsilon \geq 0$, Hoeffding's inequality implies
\begin{align}
        \Pr\left[|\overline{W} - \ex[W_j]| \ge \epsilon \right] &\le \exp \left(-\tfrac{N\epsilon^2}{2\gamma^2} \right) \label{eq:concentration-bound-1} \\
        \Pr\left[|\overline{M} - \ex[M_j]| \ge \epsilon \right] &\le \exp \left(-\tfrac{N\epsilon^2}{2\gamma^2} \right). \label{eq:concentration-bound-2}
    \end{align}
We can bound the RHS of \cref{eq:concentration-bound-1,eq:concentration-bound-2} by $0 \leq \delta < 1$ by choosing $N = \frac{2\gamma^2}{\epsilon^2}\log(1/\delta)$.
Then with probability at least $1-\delta$, we have 
\begin{align}
        \left|\frac{\overline{M}}{\overline{W}} - \frac{\mathbb{E}[M_j]}{\mathbb{E}[W_j]}\right| 
        &\leq \left| \frac{\mathbb{E}[W_j](\overline{M} - \mathbb{E}[M_j])}{\overline{W}\mathbb{E}[W_j]} \right| + \left|\frac{\mathbb{E}[M_j](\mathbb{E}[W_j] - \overline{W})}{\overline{W}\mathbb{E}[W_j]} \right|\\
        &= \frac{\epsilon}{|\overline{W}|} \left(1 + \left|\frac{\mathbb{E}[M_j]}{\mathbb{E}[W_j]}\right|\right)\\
        &\leq \frac{2\epsilon}{|\overline{W}|} \label{eq:concentration-bound-ratio-1}
    \end{align}
In the above, we have used \cref{eq:trace-estimator,eq:obs-estimator}, and that $\tfrac{\mathbb{E}[M_j]}{\mathbb{E}[W_j]}\le 1$.
Additionally, we have
\begin{align}
        \frac{1}{|\overline{W}|} \le \frac{1}{|\mathbb{E}[W_j]-\epsilon|} \le \frac{1}{|\mathbb{E}[W_j]|}\left(1+\frac{\epsilon}{|\mathbb{E}[W_j]|}\right). \label{eq:denom-G}
    \end{align}
Using \cref{eq:concentration-bound-ratio-1,eq:denom-G}, and recalling \cref{eq:rescaled-exp-vals}, we see that for the chosen $N$, with probability at least $(1-\delta)$, we have
\begin{align}\label{eq:concentration-bound-final-single}
        \left|\frac{\overline{M}}{\overline{W}} - \langle A \rangle_{\cT} \right| \le \frac{2\epsilon}{|\trace{\left[\cT(\rho)\right]}|}\left(1 + \frac{\epsilon}{|\trace{\left[\cT(\rho)\right]}|}\right).
    \end{align}
\end{proof}
\Cref{alg:est-rescaled-exp} requires a decomposition of the form \cref{eq:qpd-gen-map} as input. Given a basis, such a decomposition can be obtained classically by a linear program~\cite{temme2017Error, piveteau2020thesis}. The cost of finding the QPD scales polynomially with the Hilbert-space dimension and is thus tractable only for operations on $k$-qubit subsystems with $k \ll n$. Implementing an $n$-qubit operation therefore requires us to first approximate it as a sequence of $k$-qubit operations, and then use the QPD method for each operation in the sequence. This protocol and an analysis of its cost is the focus of the next section.
We defer a detailed discussion of how to find a suitable quasi-probability decomposition to \cref{Sec:QPDGenDecomp}. There we will also discuss some possible choices of QPD basis, and their effect on the cost of our algorithms through the factor $\gamma$.

\subsection{Implementing a sequence of operations}
\label{Sec:QPDProd}

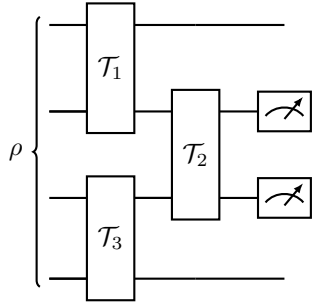
\begin{figure}[ht]
    \centering
    \begin{quantikz}
        \lstick[wires=4]{$\rho$} & \gate[2]{\cT_1} & \qw & \qw \\
         & \qw & \gate[2]{\cT_2} & \meter{} \\
         & \gate[2]{\cT_3} & \qw & \meter{} \\
         & \qw & \qw & \qw\\
    \end{quantikz}  
    \caption{Example of an operation $\cT' = \cT_1 \cT_2 \cT_3$ that we implement. 
    Here $\cT_1, \cT_2, \cT_3,$ are 2-qubit operations applied on nearest-neighbour qubits at different locations.
    In the analysis of \cref{lem:qpd-cost-repeated}, we assume for simplicity that $\cT_1 = \cT_2 = \cT_3 = \cT$.}
    \label{fig:repeated-application-map}
\end{figure}

We now explain how to implement a sequence of linear maps $\cT_1,~\cT_2,\dots,\cT_R$ where each map $\cT_i$ acts non-trivially on a constant $k$ number of qubits (see \cref{fig:repeated-application-map}). Later we will seek to implement an imaginary-time evolution by approximating it through a Trotter-Suzuki formula which gives to a sequence of $k$-local maps.
We restrict ourselves to the case where each $\cT_i$ has the form $\cT \otimes \mathcal{I}$ for some $k$-qubit map $\cT$. That is to say, $\cT_i$ it acts as $\cT$ on a $k$-qubit subsystem and as identity on all other qubits; note that each $\cT_i$ may act non-trivially on a different set of $k$ qubits. We emphasize that our algorithm and its analysis do not require the $\cT_i$'s to have this form. We make this assumption because it helps simplify notation and is sufficient for implementing imaginary-time evolution of $k$-local Hamiltonians such as the Heisenberg model. 
We present the protocol for applying a sequence of $\cT$ maps in \cref{alg:est-repeated}, and in \cref{lem:qpd-cost-repeated} we analyze its cost.

\begin{algorithm}[ht]
\caption{Estimate re-scaled expectations for repeated applications of $\cT$}\label{alg:est-repeated}
\begin{algorithmic}[1]
\Require
\Statex Operator $\mathcal{E}$ consisting of $R$ applications of $\cT$
\Statex QPD of $\cT$ as $\cT(\cdot) = \gamma \sum_i \sgn(q_i) p_i \cB_i (\cdot)$
\Statex Initial state $\rho$
\Statex Number of random samples $N$
\Statex Observable $A$
\Ensure Estimate of the re-scaled expectation value $\langle A \rangle_{\mathcal{E}}$
        \For{$j=1$ to $N$}
            \State Sample $(i_{1},...,i_{R})$ from the joint distribution $p_i^{R}$
            \For{$k=1$ to $R$}
                \If{$\cB_{i_k}$ is TP} 
                    \State Apply $\cB_{i_{k}}$
                    \State $I_{j_k}\gets1$ \Comment{Indicator variable}
                \Else
                    \State Measure in the basis corresponding to $\cB_{i_k}$
                    \State Store the outcome as $I_{j_k} \in \{0, 1\}$
                \EndIf
                \State $W_{j_k} \gets \gamma \cdot \sgn(q_{i_k}) \cdot I_{j_k}$ \Comment{Weighted outcome}
                \EndFor
        \State Measure $A$
        \State Store the outcome as $A_j$
        \State $M_j \gets \gamma^R \cdot \sgn(q_{j_1}) ... \sgn(q_{j_R}) \cdot A_j$
        \EndFor
        \State $\widetilde{\trace{\left[\mathcal{E}(\rho)\right]}} \gets \sum_{j}(\prod_{k} W_{j_k})/N$ \Comment{Empirical estimate of $\trace{\left[\mathcal{E}(\rho)\right]}$}
        \State $\widetilde{\trace{\left[A\mathcal{E}(\rho)\right]}} \gets \sum_{j} M_{j}/N$ \Comment{Empirical estimate of $\trace{\left[A\mathcal{E}(\rho)\right]}$}
        \State \Return $\widetilde{\trace{\left[A\mathcal{E}(\rho)\right]}}/\widetilde{\trace{\left[\mathcal{E}(\rho)\right]}}$
    \end{algorithmic}
\end{algorithm}

\begin{lemma}[Algorithm 2 cost]\label{lem:qpd-cost-repeated}
    Suppose we are given the QPD of a linear map $\cT$, acting non-trivially on a constant number of qubits. 
    Define the $n$-qubit map $\mathcal{E}$ as the composition of $R$ applications of $\cT$ where each $\cT$ acts non-trivially on at most $k$-qubits (see \cref{fig:repeated-application-map}). 
    Then, for a Hermitian observable $A$ with $\|A\|\le 1$, \cref{alg:est-repeated} uses $N = \mathcal{O}\left(\frac{\gamma^2}{\epsilon^2}\right)$ random circuits drawn from the distribution $p_i$ in \cref{eq:qpd-params}, to return an estimate of the rescaled expectation value $\langle A \rangle_\mathcal{E}$ within error 
    \begin{align}\label{eq:alg-2-error}
        \frac{2\epsilon}{|\trace{\left[\mathcal{E}(\rho)\right]}|}\left(1 +  \frac{\epsilon}{|\trace{\left[\mathcal{E}(\rho)\right]}|}\right).
    \end{align}
\end{lemma}
The proof of \cref{lem:qpd-cost-repeated} is similar to that of \cref{lem:qpd-cost} and can be found in \cref{app:algcosts}.

\subsection{Quasiprobability decomposition of linear maps}\label{Sec:QPDGenDecomp}
We now consider the task of obtaining the linear decomposition in \cref{eq:qpd-gen-map}, which is required as an input to \cref{alg:est-rescaled-exp,alg:est-repeated}.
Given a $k$-qubit map $\cT$ and a complete basis (consisting of $16^{k}$ linearly independent operators), the coefficients $q_i$ in \cref{eq:qpd-gen-map} can be directly solved for via a system of $16^k$ independent equations. 
The tensor product of the Endo-Benjamin-Li (EBL) basis~\cite{endo2018Practical} is an example of a complete basis for all linear maps.
It is composed of single-qubit Clifford operations (TP) and projectors (non-TP); see \cref{table:EBLBasis}.
It may be possible for a basis with fewer elements than $16^k$ to be used to obtain a QPD for a given map $\cT$.
We note that a basis containing only trace-preserving (TP) elements can only be used to express maps $\cT$ that are scalar multiples of TP operations.
However, operations such as the 
ITE operator are not simply scalar multiples of TP operations\footnote{The trace of the state resulting from applying 
 the ITE operation is dependent on the input state.}, and therefore requires the QPD basis to have at least one non-TP operation.

Different choices of basis will affect the values of the coefficients $q_i$ and hence the value of $\gamma$.
From the analysis of our algorithm in the \cref{Sec:QPDGenAlgoAnalysis}, we know that the number of samples required for estimating expectation values scales as $\mathcal{O}(\gamma^2)$.
Consequently, the sampling cost of our method can be reduced by minimizing the value of $\gamma$ in the QPD of the map of interest.
In the context of PEC, the QPD of a $k$-qubit map $\cT$ over a basis $\{\cB_i\}$, can be obtained by solving a linear program~\cite{temme2017Error, piveteau2020thesis} that minimizes $\gamma$. 
It is known that a lower bound of $\gamma$ for a HPTP map $\mathcal{T}$ is given by its diamond norm $\|\mathcal{T}\|_\Diamond$~\cite{zhao2023power,regula2021operational}
\footnote{The optimal cost $\gamma_\text{opt}$ is precisely $\|\mathcal{T}\|_\Diamond$ or $2\|\mathcal{T}\|_\Diamond$ when allowing the basis to consist of arbitrary quantum instruments, see~\cite{zhao2023power,regula2021operational}.}.

\begin{figure}[ht]
    \centering
    \includegraphics[]{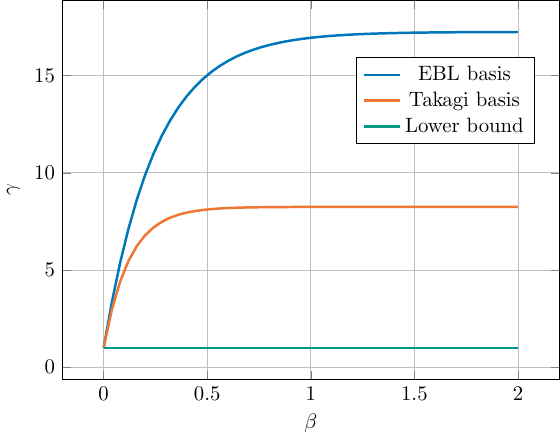}
    \caption[Scaling of QPD cost with inverse temperature]{Scaling of the QPD cost $\gamma$ with inverse temperature $\beta$ for the decomposition of $e^{-\beta H}$ using different basis sets, where $H$ is the 2-qubit Heisenberg Hamiltonian $H=-XX-YY-ZZ+II$. 
    The Takagi basis (orange) outperforms the EBL basis (blue), illustrating the utility of including entangling gates in the QPD basis set.
    The lower bound, obtained using the diamond norm~\cite{zhao2023power,regula2021operational}, is shown in teal. (See~\cref{app:optimalcost} for the analysis of optimal sampling cost and the effect of adding identity to the Hamiltonian for ITE.)}
    \label{fig:gammavbeta}
\end{figure}

As an example of the impact of the choice of basis, we consider the ITE operator $e^{-\beta H}$ of the $2$-qubit Heisenberg Hamiltonian $H = -XX -YY -ZZ+II$\footnote{See \cref{app:optimalcost} for an explanation of the added identity term.}.
In \cref{fig:gammavbeta}, we show the trends of $\gamma$ as a function of $\beta$ for two different choices of basis. 
We find that the ITE operator for the $2$-qubit Heisenberg Hamiltonian can be decomposed into the Takagi basis (\cref{table:takagi}) with a lower value of $\gamma$ than the EBL basis, where the Takagi basis value is closer to the lower bound given by the diamond norm (see \cref{fig:gammavbeta}).
This reduced value of $\gamma$ can be attributed to the presence of entangling operations in the Takagi basis.
Since the EBL basis for $2$-qubits is simply the (uncorrelated) Cartesian product of two single-qubit basis sets, it has a high sampling overhead for simulating entangling operations in expectation.
This is apparent in the case of the CNOT operation~\cite{endo2018Practical,takagi2021optimal}, which has $\gamma_{\mathrm{CNOT}} = 9$ for the $2$-qubit EBL basis but $\gamma_{\mathrm{CNOT}} = 1$ when decomposed into the Takagi basis.
We remark that although the Takagi basis permits a decomposition of the ITE operator in this case, it is not a complete basis for general HP maps on a 2-qubit system.
The Takagi basis is only known to be complete for the set of all $2$-qubit CPTP maps~\cite{takagi2021optimal}.
A more detailed analysis of QPD and choices of basis can be found in~\cref{app:choi_qpd}.

The problem of obtaining a QPD for a given map requires careful consideration of two important factors: the choice of basis set used for the decomposition and the value of $\gamma$ obtained from the QPD.
In general, the QPD corresponding to the optimal value of $\gamma$ could be a linear combination of entangling operations or general quantum instruments that need to dilate to a bigger Hilbert space, which may be expensive (in gate depth) to implement directly on a quantum device. 
This tradeoff between the generality of the basis and the optimality of the decomposition~\cite{endo2018Practical,endo_mitigating_2019,takagi2021optimal} has been well studied for the specific case of QPDs for CPTP maps, with several results on optimizing~\cite{piveteau2022Quasiprobability,piveteau2020thesis} the choice of basis for different use-cases.
However, the solution to this problem for general HP maps, including NTP maps such as ITE which commonly occur in physics, is yet (to the best of our knowledge) unknown.
We provide some discussions in~\cref{app:choi_qpd}.
Additionally, the problem of finding a sample-efficient yet hardware-friendly basis for decomposing maps acting on systems larger than $2$-qubits quickly becomes intractable (since the size of the associated classical optimization problem grows polynomially with the dimension of the Hilbert space) and remains a challenge even for the case of decomposing CPTP maps.

\section{QPD implementation of imaginary-time evolution}\label{Sec:QPDforQITE}

We now use the algorithms developed in \cref{Sec:QPDGeneral} to implement imaginary time evolution $e^{-\beta H}$ for a local Hamlitonian $H$
\begin{align}\label{eq: localHam}
    H = \sum_{l=1}^L H_l,
\end{align}
presented as a sum of $L=O(\poly(n))$ local terms $H_l$ acting non-trivially on at most a constant $k$ number of qubits. 
We may assume without loss of generality that $H \ge 0$ and $\|H\| \le 1$ where $\|\cdot\|$ denotes the operator norm. 
Our goal is to implement the operator $e^{-\beta H}$ for real $\beta > 0$. 
More precisely,
for a given initial state $|\psi\rangle$, we want to estimate expectation values of a given observable $A$ with respect to the quantum state $|\psi_\beta \rangle \propto e^{-\beta H}|\psi\rangle$. 
Some basics of ITE and the optimal decomposition of ITE are presented in~\cref{app:ITE} and \cref{app:optimalcost}, respectively.
 
\subsection{Warm-up example: 2-qubit Heisenberg model}
We first consider the simple case of implementing the ITE of a 2-qubit Heisenberg model 
\begin{align}\label{eq:2qHeisHamiltonian}
    H = -XX - YY - ZZ
\end{align}
Since the basis $\{\cB_i\}$ in \cref{table:EBLBasis} is known to be complete for all linear maps on single-qubits, we use the 256-element Cartesian product set $\{\cB_i\} \times \{\cB_j\}$ as the basis for decomposing the $2$-qubit ITE operator $e^{-\beta H}$ into a QPD.
In \cref{HeisExpt}, we simulate \cref{alg:est-rescaled-exp} for estimating the energy of the $2$-qubit Heisenberg model when it is initialized in the state $\ket{00}$ and undergoes ITE for increasing values of $\beta = 0.01t$, where $t$ is the Trotter step.
We note that the simulation is noiseless and there is no Trotter error because $H$ commutes with itself.
Therefore, the only error in the simulation is due to the QPD sampling.

Additionally, we implement the $2$-qubit ITE operator $e^{-\beta H}$ using \cref{alg:est-rescaled-exp} on a 5-qubit superconducting processor (\texttt{ibm\_manila}).
In this case, we use the noisy native gates on the device as the basis for QPD using noise characterization information.
In particular, we take the noisy basis $\{\cB'_i\}$ to be $\{\mathcal{N} \circ \cB_i\}$ where $\mathcal{N}$ is the noise channel used to model the noise in the device, based on calibration data provided by IBM \cite{ibmq_calibration}.

In the simulation and demonstration on hardware, we used 400, 800, 3200, and 25600 QPD samples for Trotter steps 1, 2, 3, and 4, respectively.
For each sample, 512 shots of the corresponding circuit were performed (or simulated).
For each Trotter step, this process was repeated 10 times; the data in \cref{HeisExpt} show the mean estimate and the standard deviation.

\begin{figure}[ht!]%
    \centering
    \includegraphics[scale=.95]{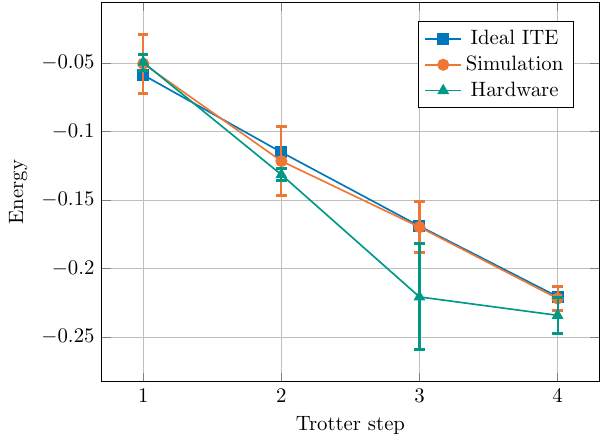}%
    
    \caption{Energy estimation using ITE. We estimate the energy $\langle \psi_t | H | \psi_t \rangle / \langle\psi_t|\psi_t\rangle$, where $H$ is the 2-qubit Heisenberg Hamiltonian and $|\psi_t\rangle = (e^{-0.01H})^t|00\rangle$ is the imaginary-time evolved state with $t$ Trotter steps. The blue, orange, and teal points show the values for exact, simulated, and hardware ITE, respectively, where the simulation (resp., demonstration) implements \cref{alg:est-repeated} on classical (resp., quantum) hardware.
    }
    \label{HeisExpt}
\end{figure}

The agreement between the demonstration on noisy hardware and the noiseless (classical) simulation indicates that our method can be made noise-resilient when information about the noise in a device is available.  
We note that the demonstration data for the third Trotter step is significantly worse than for the other Trotter steps.
A possible explanation for this is that the device noise drifts over time, and so the noisy basis $\{\cB_i'\}$ may become inaccurate, leading to bias in the results.
This is a particular issue for demonstrations such as ours where we ran many different circuits, each with a small number of shots, which is much slower than running a small number of circuits with many shots each.
The results could also possibly be improved by using more QPD samples (compare the size of the (simulated data) error bars of the third and fourth Trotter steps in \cref{HeisExpt}).

\subsection{QPD implementation of Trotterized imaginary time evolution}

We now define $\cT(\cdot)$ to be the imaginary-time evolution for the $n$-qubit Hamiltonian in \cref{eq:productapproxHam}
\begin{align}
    \cT(\cdot) = e^{-\beta H}(\cdot)e^{-\beta H}.
\end{align}
We cannot apply \cref{alg:est-rescaled-exp} directly to $\cT(\cdot)$ since the task of obtaining the quasiprobability decomposition over an $n$-qubit basis set $(|\{\cB_i\}| = \mathcal{O}(16^n))$ by solving a linear program can quickly become intractable for large $n$. Accordingly, we first decompose the $n$-qubit operation $e^{-\beta H}$ to a sequence of $k$-qubit operations by using the $1^{\mathrm{st}}$-order Trotter approximation: 
\begin{align}\label{eq:productapproxHam}
    e^{-\beta H} &\approx \left( e^{-\beta H_1/r}e^{-\beta H_2/r}\ldots e^{-\beta H_L/r} \right)^r.
\end{align}
Since each $H_l$ is $k$-local, we obtain the quasiprobability decomposition of each local ITE operation $e^{-\beta/r H_l}$ separately and implement a Trotterized ITE using \cref{alg:est-repeated}. 
Let $\cT_l(\cdot)$ be the ITE operation for the $k$-local Hamiltonian $H_l$
\begin{align}
    \cT_l(\cdot) = e^{-\beta H_l/r} (\cdot) e^{-\beta H_l/r}
\end{align}
We require the QPDs of each of the linear maps $\cT_1, \ldots , \cT_L$ into a set of operations $\{\cB_i\}$ as in \cref{eq:qpd-gen-map}. Suppose that we have computed these decompositions with the associated parameters $\gamma_1, \ldots , \gamma_L$ respectively. 
Then, we define the $n$-qubit map $\mathcal{P}$ as the composition of the maps $(\cT_1,...,\cT_L)$ corresponding to the $1^{\mathrm{st}}$ order Trotter approximation in \cref{eq:productapproxHam}.
For a Hermitian observable $A$ with $\|A\|\le 1$, it follows from \cref{lem:qpd-cost-repeated} that implementing $\mathcal{P}$ using \cref{alg:est-repeated} returns an estimate of the re-scaled expectation value $\langle A \rangle_{\mathcal{P}}$ within additive error $\epsilon \cdot (\tr{\left[\mathcal{P}(\rho)\right]})^{-1}$ using $\mathcal{O}(\gamma_1^{2r}...\gamma_L^{2r}\epsilon^{-2})$ circuits. 

In our implementation we only consider a translation-invariant Hamiltonian, i.e., each local interaction term in $H$ is the same. As a result, each of the exponentials $e^{-\beta H_l /r}$ has the same decomposition in a given basis with an associated parameter $\gamma$. Now, the Trotter approximation causes the expectation value $\langle A \rangle_{\mathcal{P}}$ to differ from the target imaginary-time expectation $\langle A \rangle_{\cT}$, but this difference can be bounded using known results~\cite{childs2021theory}. We show in \cref{app:ite-cost-proofs} that for any observable $A$,
\begin{align}
    |\trace{[A\cT(\rho)]} - \trace{[A\mathcal{P}(\rho)]}| &= \mathcal{O}\left(\frac{\beta^2 L^2 }{r} e^{\beta L/r} \right) \label{eq:err-obs-trot}  \\
    |\trace{[\cT(\rho)]} - \trace{[\mathcal{P}(\rho)]}| &= \mathcal{O}\left(\frac{\beta^2 L^2 }{r} e^{\beta L/r} \right) \label{eq:err-tr-trot}
\end{align} 
again assuming $\|A\| \le 1$.
Note that we can choose 
\begin{align}
r = \mathcal{O}\left(\tfrac{\beta^2 L^2}{\epsilon}\right) \label{eq:trotter-num}
\end{align}
to make the error on the RHS at most $O(\epsilon)$.
We know from \cref{lem:qpd-cost-repeated} that \cref{alg:est-repeated} can estimate the expectation value $\langle A \rangle_{\mathcal{P}} = \frac{\trace{(A\mathcal{P}(\rho))}}{\trace{(\mathcal{P}(\rho))}}$ to within error 
\begin{align}\label{eq:err-estimation-ite}
    \frac{2\epsilon}{|\trace{\left[\mathcal{P}(\rho)\right]}|}\left(1 + \frac{\epsilon}{|\trace{\left[\mathcal{P}(\rho)\right]}|}\right)
\end{align}
using 
\begin{align}
    \mathcal{O}\left(\frac{\gamma^{2Lr}}{\epsilon^2}\right)
\end{align}
circuits where each circuit consists of at most $\mathcal{O}(Lr)$ gates. Now, the algorithm is suitable only when the normalization $\trace{(\mathcal{P}(\rho))}$ is much larger than the target precision $\epsilon$. In that case, for the choice of $r$ in \cref{eq:trotter-num}, it follows that
\begin{align}
    |\langle A \rangle_{\mathcal{T}} - \langle A \rangle_{\mathcal{P}} | \le \mathcal{O}(\epsilon) \; ,
\end{align}
and moreover, the error in estimating $\langle A \rangle_{\mathcal{P}}$ via \cref{eq:err-estimation-ite} can be upper-bounded by
\begin{align}
    \frac{4\epsilon}{\trace{\left[\cT(\rho)\right]}}\left(1 + \frac{4\epsilon}{\trace{\left[\cT(\rho)\right]}}\right).
\end{align}
Even though the cost as stated scales as $\gamma^{2Lr}$, we note that the factor $\gamma$ depends on the local exponential $e^{-\beta H_i/r}$ and is therefore likely to decrease as we increase the Trotter number $r$.
Moreover, our error bounds above can likely be further tightened, in particular because we appeal to standard additive-error bounds for Trotter formula which may not be well-suited for non-trace-preserving imaginary-time evolution. It may also be possible to improve the bounds by using multiplicative-error estimates for Trotter-Suzuki formula~\cite{bravyi2017polynomial,childs2021theory}. We leave a detailed analysis of the complexity of our algorithm to future work.

\subsubsection{Estimating thermal expectation values}

We now show how our QPD implementation of imaginary-time evolution can be used to compute expectation values of observables in the Gibbs state.
Recall that the Gibbs state is defined to be
\begin{equation}
    \rho_\beta = \frac{e^{-\beta H}}{\mathcal Z_\beta},
\end{equation}
where $H$ is the Hamiltonian of the quantum system, $\beta$ is the inverse temperature and $\mathcal Z_\beta = \trace {(e^{-\beta H})}$ is the partition function. Preparing the Gibbs state and estimating thermal expectation values for general local Hamiltonians are known to be difficult tasks, and therefore unlikely to admit efficient algorithms in general~\cite{bravyi2022quantum}. 

Instead of directly preparing the (mixed) Gibbs state, we seek to prepare a \emph{thermal pure quantum state} (TPQ) \cite{sugiura2013canonical}. The latter is a suitably constructed $n$-qubit pure state $|\psi \rangle$ which, with high probability, can be used to estimate the Gibbs state expectation values for a set of observables, up to a tunable precision $\epsilon$. 
Formally, the TPQ state for a given set of observables $\{ \hat O \}$, is defined as a pure state $|\psi \rangle$ that satisfies
\begin{equation}\label{defn-tpq-state}
    \Pr [
    \langle \psi | \hat O | \psi \rangle 
    - \trace (\rho_\beta \hat O)
    \geq \epsilon
    ]
    \leq f(N, \epsilon),
\end{equation}
where $\lim_{N \rightarrow \infty} f(N, \epsilon) = 0$.~\cite{sho2012thermal,sho2013canonical}
It is known that applying imaginary-time evolution to a state chosen uniformly at random from the Haar measure gives a TPQ state~\cite{richter2021simulating,powers2023exploring,sho2013canonical}. Since Haar-random states cannot be efficiently prepared on a quantum computer, we follow the approach of Ref.~\cite{coopmans2023predicting} where the initial Haar-random state is replaced by a random 3-design. Specifically, the authors in Ref.~\cite{coopmans2023predicting} apply a random Clifford operator to a computational basis state followed by imaginary-time evolution, preparing the state,
\begin{equation}\label{eq:clifford-tpq}
    | \psi_\beta \rangle =
    \frac{
        e^{-\beta H/2} U | 0 \rangle
    }{
        \sqrt{ \langle 0 | U^\dagger e^{-\beta H} U | 0 \rangle}
    },
\end{equation}
where $U$ is the random Clifford operator, and the state $| \psi_\beta \rangle$ is in normalized form. Note that a random $n$-qubit Clifford operator can be implemented using $\mathcal O (n^2)$ one- and two-qubit gates~\cite{bravyi2021haramard}, a significant gain over the exponential number of such gates required to generate a Haar-random unitary. Leveraging the 3-design property of the Clifford group, it can be shown ~\cite{coopmans2023predicting,watts2023Quantum} that the above state (Eqn~\ref{eq:clifford-tpq}) is a valid TPQ state satisfying the definition in \cref{defn-tpq-state}. 

Ref.~\cite{coopmans2023predicting} used algorithms based on quantum signal processing to implement the ITE operation to prepare the Clifford-random TPQ. Here, we implement ITE using our quasi-probabilistic approach based on applying \cref{alg:est-repeated} to a Trotterized approximation of ITE. We note that preparing TPQ states is a highly suitable application of \cref{alg:est-repeated} because the TPQ state error rapidly approaches zero as $n$ becomes large~\cite{coopmans2023predicting}.
In addition, for a $k$-local $n$-qubit Hamiltonian all that is needed is the QPDs of the local terms, which is much more tractable to compute than the QPD of the entire $n$-qubit Hamiltonian.

The method we outline here can be used to estimate Gibbs-state expectation values for any $k$-local Hamiltonian at a given temperature. A main advantage of our method is that we do not require \emph{any} ancilla qubits. Similar to quantum algorithms for Gibbs-state preparation based on block-encoding techniques \cite{poulin09sampling,chowdhury2017quantum,vanApeldoorn2020quantumsdpsolvers,gilyen2019singular}, we expect our method to require an exponentially long running-time in the worst-case. However, the above algorithms for state preparation require at least $n$ ancilla qubits to prepare the Gibbs state of an $n$-qubit Hamiltonian. Using the TPQ method with block-encoding approaches reduces the number of ancila qubits from $n$ to a constant, as was shown in Ref.~\cite{coopmans2023predicting}. Using our quasiprobabilistic approach eliminates the need for any ancillary qubits. In this respect, our method is comparable to the dissipative quantum Gibbs sampler of Ref.~\cite{zhang2023dissipative}. 
In general, there have been significant developments in quantum algorithms that prepare Gibbs states by simulating Lindbladian evolution \cite{Chen2023QuantumTS,rall2023rounding,shtanko2023preparingthermalstatesnoiseless,ding2025gibbs}. With the exception of Ref.~\cite{shtanko2023preparingthermalstatesnoiseless}, these algorithms are designed for fault-tolerant quantum computers and require non-trivial subroutines. 

\paragraph{Numerical simulations.}
We test the performance of our technique in this case by simulating the preparation of TPQ states using \cref{alg:est-repeated} and comparing their performance with that of TPQ states prepared using ideal ITE.
Specifically, we consider TPQ states of the form $|\psi_i\rangle = e^{-0.02H}U_i|0\rangle^{\otimes n}$, where $H$ is the 1D Heisenberg Hamiltonian on $n\in\{4,\ldots,8\}$ qubits, and $U_i$ is a randomly-chosen Clifford unitary. 
We generate 10 such random states and average the error in the thermal expectation value, $|\langle \psi_i | O_j | \psi_i \rangle - \trace(|\psi_\beta\rangle\langle\psi_\beta|O_j)|$, over 30 randomly chosen Pauli observables, $O_j$.
For each simulated TPQ state, we use $N \in \{1024, 9400, 51200, 409600, 1638400\}$ QPD samples for $n \in \{4, 5, 6, 7, 8\}$ qubits respectively.
The results are shown in \cref{fig:tpq-demo} and show good agreement between our simulation and states prepared using ideal ITE.

\section{Conclusion} \label{sec:conc}

In this work, we propose using quasiprobability decompositions to implement imaginary time evolution, a useful primitive in many quantum algorithms.
This ``algorithmic'' use of quasiprobability decompositions inherits the error mitigation properties of related techniques such as probabilistic error cancellation.
Our method also extends to implementing linear Hermiticity-preserving maps, of which imaginary time evolution is one example.
We provide algorithms for estimating the expectation values of observables with respect to states evolved under one or more applications of a given map.
Using our algorithms, we can estimate the expectation value of an observable for imaginary-time evolved states.
This has many applications including preparing thermal pure quantum states, which can be used to estimate the Gibbs state expectation value of a set of observables.
We simulate the performance of our method for this task, and find good agreement with the ideal imaginary time evolution operation.
Our method is well-suited to near-term hardware, as it requires only 1- and 2-qubit gates and measurement, and the quasiprobability decomposition can be performed over a device's native (noisy) gateset.
As a proof-of-principle demonstration, we implement imaginary time evolution using a superconducting qubit processor (provided by IBM), and use this to estimate the ground state energy of the 2-qubit Heisenberg Hamiltonian.

A natural next step will be to extend our proof-of-principle demonstration on hardware to Hamiltonians on more than two qubits.
Although this will entail increased sampling costs, when one considers a more specific use-case, namely of thermal pure quantum states to estimate (thermal) expectation values, the error in the expectation value decays rapidly with the number of qubits, a promising sign for application of our method on near-term devices.
There are two hardware-level optimizations that would make the methods presented in this work more efficient.
First, in our demonstration using IBM hardware, if a mid-circuit measurement the incorrect result, i.e., post-selection is unsuccessful, then we still execute the remainder of the circuit because of the limitations of the measurement-based feedforward.
If instead one has the ability to apply quantum circuits based on the results of mid-circuit measurements, then this waste of resources can easily be avoided.
Second, sampling all of the required circuits in software and then loading each circuit individually into hardware is highly inefficient.
If instead the sampling is implemented at the level of the field-programmable gate-array control system (as was recently demonstrated for randomized compiling~\cite{fruitwala2024hardware}), then we expect that the overall runtime of the circuits will be significantly reduced. 

\section*{Acknowledgments}

The authors thank David Gosset, Joel Klassen, Raymond Laflamme, and Christophe Piveteau for enlightening discussions. 
This research was enabled in part by support provided by \href{https://computeontario.ca}{Compute Ontario} and the \href{https://alliancecan.ca}{Digital Research Alliance of Canada}. 
This work was completed while AR was a doctoral candidate at the University of Waterloo.
Part of this research was conducted while AC was affiliated with the University of Waterloo. NC and AC were partly supported by the Natural Sciences and Engineering Research Council of Canada (NSERC) under grant RGPIN-2019-04198.
Research at Perimeter Institute
is supported in part by the Government of Canada through
the Department of Innovation, Science and Economic Development Canada and by the Province of Ontario through the Ministry of Colleges and Universities.

\clearpage
\newpage
\bibliographystyle{quantum}

\appendix

\section{Choi representation and decompositions for linear maps}\label{app:choi_qpd}

The Choi representation~\cite{choi1975completely,jamiolkowski1972linear} (a.k.a Choi–Jamiołkowski isomorphism) captures many important properties of linear maps, hence plays a vital role in quantum information science~\cite{watrous2018theory}.
The Choi representation $C(\cT)$ of a linear map $\cT: \mathbb{B}_{\mathcal{H}_1} \to \mathbb{B}_{\mathcal{H}_2} $ is defined as
\begin{equation}\label{eq:choi}
    C(\cT) := \cT\otimes I (\ket{\psi_+}\bra{\psi_+}),
\end{equation}
where $\mathbb{B}_{\mathcal{H}_1}$ and $\mathbb{B}_{\mathcal{H}_2}$ are bounded operators on the Hilbert space $\mathcal{H}_1$ and $\mathcal{H}_2$ respectively, 
$\ket{\psi_{+}} = \sum_{i=1}^d\ket{ii}\in \mathcal{H}_1 \otimes \mathcal{H}_1$ is an unnormalized maximally entangled state. 
The map $\cT$ is Hermiticity-preserving (HP) iff $C(\cT)$ is a Hermitian matrix. 
The map $\cT$ is completely positive (CP) iff $C(\cT) \ge 0$ (i.e. $C(\cT)$ is positive semi-definite)~\cite{watrous2018theory}.

Choi representation $C(\cT)$ preserves the linear combination of maps, meaning that
\[
C(\cT) = \sum_i q_i C(\cB_i)
\]
if $\cT = \sum_i q_i \cB_i$.
Therefore the QPD decomposition can be done within the Choi representation. 
For HP linear maps, the coefficients $q_i$ can be chosen to be real numbers with a Hermitian basis. 
Note that CP maps are a subset of all HP maps.
Since ITE is not only HP, it is also CP. Therefore, the basis elements, $\{\cB_i\}$, that we use for demonstrations (see~\cref{table:EBLBasis}) are all CP, the coefficients $q_i$ are all real. 

The basic idea of QPD is decomposing the desired operation $\cT$ as a linear combination of other physically realizable operations $\{\cB_j\}$, then sampling each operation $\cB_j$ according to the probability $p_j = q_j / \sum_i |q_i|$ calculated from coefficients $\{q_i\}$.
The ``free operations'', i.e. $\gamma = \sum_i |q_i| = 1$, are the convex combination $\text{Cov}(\cB_j)$ of $\{\cB_j\}$.
For example, the free operations for the 2 fixed operations $\mathcal{B}_1$ and $\mathcal{B}_2$ would be the line segment before them, and the operation they can express is any operation on the line. 
The choice of $\{\cB_j\}$ is critical regarding the implementation of QPD.

One can choose the physical operations $\{\cB_j\}$ to be a fixed finite set, then the QPD method can achieve all linear combinations of the basis. For a $n$-qubit linear operation, the degree of freedom is $4^{2n}$.
Hence, a linearly independent set of bases with cardinality $4^{2n}$ can express all linear operations for $n$ qubits.
For example, \cref{table:EBLBasis}, proposed in~\cite{endo_mitigating_2019}, is complete for single-qubit operations, and the tensor product of these bases is also complete for more qubits.
The basis in~\cref{table:EBLBasis} comprises 10 unitaries and 6 measurements. The unitary operators are generated by $S$ gate and Hardmard gate $H$, and are thus Clifford. The measurements are generated by $S$, $H$ and projection $P_0$. They are all relatively easy and directly implementable on hardware and relatively easy concerning classical simulation.

\begin{table}[ht!]
    \centering
    \begin{tabular}{|l c l c l|}
        \hline
         $[1]$ & & & &   \\
         $[\sigma^X]$ & & & $=$ & $[H][S]^2[H]$ \\
         $[\sigma^Y]$ & & & $=$ & $[H][S]^2[H][S]^2$ \\
         $[\sigma^Z]$ & & & $=$ & $[S]^2$ \\
         $[R_X]$ & $=$ &$[\frac{1+\iu \sigma^X}{\sqrt{2}}]$ & $=$ & $[H][S]^3[H]$ \\
         $[R_Y]$ & $=$ &$[\frac{1+\iu \sigma^Y}{\sqrt{2}}]$ & $=$ & $[S][H][S]^3[H][S]^3$ \\
         $[R_Z]$ & $=$ &$[\frac{1+\iu \sigma^Z}{\sqrt{2}}]$ & $=$ & $[S]^3$ \\
         $[R_{YZ}]$ & $=$ &$[\frac{\sigma^Y+\sigma^Z}{\sqrt{2}}]$ & $=$ & $[H][S]^3[H][S]^2$ \\
         $[R_{ZX}]$ & $=$ &$[\frac{\sigma^Z+\sigma^X}{\sqrt{2}}]$ & $=$ & $[S]^3[H][S]^3[H][S]^3$ \\
         $[R_{XY}]$ & $=$ &$[\frac{\sigma^X+\sigma^Y}{\sqrt{2}}]$ & $=$ & $[H][S]^2[H][S]^3$ \\
         $[\pi_X]$ & $=$ &$[\frac{1+\sigma^X}{2}]$ & $=$ & $[S][H][S][H][P_0][H][S]^3[H][S]^3$ \\
         $[\pi_Y]$ & $=$ &$[\frac{1+\sigma^Y}{2}]$ & $=$ & $[H][S]^3[H][P_0][H][S][H]$ \\
         $[\pi_Z]$ & $=$ &$[\frac{1+\sigma^Z}{2}]$ & $=$ & $[P_0]$ \\
         $[\pi_{YZ}]$ & $=$ &$[\frac{\sigma^Y+\iu\sigma^Z}{2}]$ & $=$ & $[S][H][S][H][P_0][H][S][H][S]^3$ \\
         $[\pi_{ZX}]$ & $=$ &$[\frac{\sigma^Z+\iu\sigma^X}{2}]$ & $=$ & $[H][S]^3[H][P_0][H][S][H][S]^2$ \\
         $[\pi_{XY}]$ & $=$ &$[\frac{\sigma^X+\iu\sigma^Y}{2}]$ & $=$ & $[P_0][H][S]^2[H]$ \\
         \hline
    \end{tabular}
    \caption{\textbf{EBL basis}~\cite{endo2018Practical}. A complete basis for single-qubit linear maps that contains 10 trace-preserving elements ($[1]$ to $[R_{XY}]$) and 6 non-trace-preserving elements ($[\pi_X]$ to $[\pi_{XY}]$), where we use the notation $[U](\cdot)\coloneqq U(\cdot)U^\dagger$. The map $[P_0]$ refers to projection onto the state $\ket{0}$.
    }
    \label{table:EBLBasis}
\end{table}

The number of degrees of freedom in the problem can be reduced with constraints on the map $\cT$. For instance, knowing $\cT$ is TP reduces the degrees of freedom to  $d^4 - d^2 + 1$. In~\cite{takagi2021optimal}, the author proposed a new basis set (see~\cref{table:takagi}) for $2$-qubit CPTP maps, which reduced the number of basis from 256 to 241. 
Note that, as shown in~\cref{fig:gammavbeta}, decomposing the imaginary time evolution according to the EBL basis has a higher sampling cost compared to the Takagi basis. 
It is not only because the Takagi basis contains more entangling operations. As mentioned above, the geometry of the basis sets also affects the cost.

\begin{table}[ht!]
    \centering
    \begin{tabular}{|c|c|}
    \hline
    $\cB_{1}$--$\cB_{169}$ & $\{\cB_i\}_{i=1}^{13}\otimes \{\cB_i\}_{i=1}^{13}$\\
    $\cB_{170}$--$\cB_{178}$ & $\mathcal{CX}$ + conjugation with $\mK_{1,2},\mK_{1,2}^\dagger$ \\
    $\cB_{179}$--$\cB_{187}$ & $\mX_1\circ\mathcal{CX}\circ\mX_1$ + conjugation with $\mK_{1,2},\mK_{1,2}^\dagger$ \\
    $\cB_{188}$--$\cB_{196}$ & $\mathcal{CS}$ + conjugation with $\mK_{1,2},\mK_{1,2}^\dagger$ \\
    $\cB_{197}$--$\cB_{205}$ & $\mathcal{CH}$ + conjugation with $\mK_{1,2},\mK_{1,2}^\dagger$ \\
    $\cB_{206}$--$\cB_{214}$ & $\mathcal{C_H X}$ + conjugation with $\mK_{1,2},\mK_{1,2}^\dagger$ \\
    $\cB_{215}$--$\cB_{223}$ & $\mathcal{C X}\circ\mH_1$ + conjugation with $\mK_{1,2},\mK_{1,2}^\dagger$ \\
    $\cB_{224}$--$\cB_{226}$ & $\mathcal{S}_W$ + conjugation with $\mK_2,\mK_2^\dagger$ \\
    $\cB_{227}$--$\cB_{232}$ & $i\mathcal{S}_W$ + conjugation with $\mK_{1,2}, \mK_{2}^\dagger$ \\
    $\cB_{233}$--$\cB_{241}$ & $\mathcal{S}_W\circ\mH_1$ + conjugation with $\mK_{1,2},\mK_{1,2}^\dagger$ \\
    \hline
    \end{tabular}
    \caption[Universal basis for QPD of $2q$ CPTP maps (Takagi basis)]{\textbf{Takagi Basis}~\cite{takagi2021optimal}. A complete basis for 2-qubit CPTP maps. 
    $\{\cB_i\}_{i=1}^{13}$ are the first 13 elements of the EBL basis (\cref{table:EBLBasis}).
    $\mathcal{CX}$, $\mathcal{CS}$, $\mathcal{CH}$, $\mathcal{C_H X}$ are channel versions of CNOT, controlled-phase, controlled-Hadamard, and NOT controlled with $\pm 1$ eigenstates of the Hadamard gate, respectively.
    $\mathcal{K}_i$ is the channel version of $K \coloneqq SH$ acting on qubit $i$. 
    $\mS_W$ and $i\mS_W$ are channel versions of the SWAP and iSWAP gates. 
    $\mU$ conjugated by $\mV$ denotes $\mV^\dagger \circ \mU \circ \mV$ and ``$\mU$ + conjugation with $\mK_{1,2}, \mK_{1,2}^\dagger$'' collects the nine conjugations of $\mU$ by
    $I_{12}$, $\mK_1$, $\mK_2$, $\mK_1^\dagger$,
    $\mK_2^\dagger$, $\mK_1\otimes\mK_2$,
    $\mK_1\otimes\mK_2^\dagger$,
    $\mK_1^\dagger\otimes\mK_2$, and
    $\mK_1^\dagger\otimes \mK_2^\dagger$.}
    \label{table:takagi}
\end{table}

On the other hand, one can also try decomposing $\cT$ to more general sets of physical operations. 
It is known that any HPTP map $\mathcal{T}$ can be decomposed to a linear combination of two CPTP maps $\mathcal{T} = q_1 \mathcal{T}_1 - q_2 \mathcal{T}_2$~\cite{jiang2021physical}. Similarly, a HP map can be decomposed by the subtraction of two CP maps. 
This also allows us to implement QPD without fixing a set of basis operations~\cite{jiang2021physical}. Similarly, allowing general quantum instruments can also empower us to implement QPD for broader classes of maps. 
In~\cite{zhao2023power}, authors proved the optimal cost of QPD will achieved by allowing even one quantum instrument.
The optimal sampling cost $\gamma$ for implementing with measurement-controlled post-processing (a generalized QPD method) equals the diamond norm $\|\cT\|_\Diamond$ of the desired map $\cT$~\cite{zhao2023power}.
We analyze the optimal decomposition of ITE in~\cref{app:optimalcost}.

\section{Basics of imaginary time evolution}~\label{app:ITE}

The imaginary time evolution without renormalization is a CP non-TP map, 
$$\mathcal{T}(\rho) = e^{-\beta H}\rho e^{-\beta H},$$
where $\beta>0$ and $\rho\in \mathcal{D}(\mathcal{H}^d)$.

Diagonalizing the ITE operator 
\begin{equation}\label{eq:diag}
e^{-\beta H} = U\text{diag}[e^{-\beta \lambda_i}]U^\dag,
\end{equation}
where $\text{diag}[e^{-\beta\lambda_i}]$ is a diagonal matrix and $\{\lambda_i\}$ are eigenvalues of $H$ with the increasing order ($\lambda_0$ is the ground state energy). It is a convention to assume that $H$ is positive semi-definite (i.e. $\lambda_i \ge 0$). However, we need to be slightly careful when working with non-unitary evolution.

From~\cref{eq:diag}, when $\beta \to \infty$, if $\lambda_0 = 0$, the ITE operator is the projection onto the space of the ground states; if $\lambda_0 > 0$, the limit $\lim_{\beta\to\infty} e^{-\beta H}$ will become a zero matrix; if $H$ has negative eigenvalues then $\lim_{\beta\to\infty} e^{-\beta H}$ is an unbounded operator that projects onto the ground space and multiplies the unbounded scalar $e^{-\beta\lambda_0}$.

\section{Optimal sampling cost for ITE}\label{app:optimalcost}

If we allow ourselves the ability to optimize over the choice of the basis elements $\{\cB_i\}$, we can upper bound the cost of implementing $e^{-\beta H}$ in terms of $\beta$ and the operator norm of $H$.
The diamond norm $\|\mathcal{T}\|_\Diamond$ of the HP map $\mathcal{T}$ gives a lower bound for the QPD sampling cost, also is the optimal cost of a generalized QPD implementation~\cite{zhao2023power}. In the following, we analyze the diamond norm for the imaginary time evolution.

The imaginary time evolution without renormalization is a CP non-TP map, 
$$\mathcal{T}(\rho) = e^{-\beta H}\rho e^{-\beta H},$$ 
where $\beta$ is fixed and $\beta>0$ and $\rho\in \mathcal{D}(\mathcal{H}^d)$.
The diamond norm of $\mathcal{T}$ is
\[
\da{\mathcal{T}} := \sup_{n\ge 1} \|\mathcal{T}\otimes\mathbb{I}_n\|_{tr},
\]
where $\|\mathcal{T}\otimes\mathbb{I}_n\|_{tr} = \max_\rho \|\mathcal{T}\otimes\mathbb{I}_n (\rho)\|_{tr}$, and $\rho \in \mathcal{D}(\mathcal{H}^d\otimes\mathcal{H}^n)$. Therefore,
\[
\da{\mathcal{T}} = \sup_{n\ge 1} \max_\rho \|\mathcal{T}\otimes\mathbb{I}_n (\rho)\|_{tr} = \max_{\ket{\psi}}\trn{\mathcal{T}\otimes\mathbb{I}_d(\ket{\psi}\bra{\psi})},
\]
where $d$ is the dimension of the system, and now $\rho \in \mathcal{D}(\mathcal{H}^d\otimes\mathcal{H}^d)$.
The second equality is due to the convexity of trace norm and Schmidt decomposition.

Perform the spectral decomposition of the Hamiltonian $H$, $H = U\Sigma U^\dag.$
For imaginary time evolution $\mathcal{T}(\rho) = e^{-\beta H}\rho e^{-\beta H}$, its diamond norm is
\begin{align}
    \da{\mathcal{T}} &= \max_{\ket{\psi}}\trn{[(Ue^{-\beta\Sigma}U^\dag)\otimes\mathbb{I}_d](\ket{\psi}\bra{\psi})[(Ue^{-\beta\Sigma}U^\dag)\otimes\mathbb{I}_d]} \nonumber\\
    &= \max_{\ket{\psi}} \trace(\sqrt{\ket{\phi}\bra{\phi}\ket{\phi}\bra{\phi}}) \nonumber\\
    & = \max_{\ket{\psi}} a_\phi
\end{align}
where $\ket{\phi}:= [(Ue^{-\beta\Sigma}U^\dag)\otimes\mathbb{I}_d]\ket{\psi}$ (note that $\ket{\phi}$ is unnormalized since $\mathcal{T}$ is not trace-preserving) and $a_\phi := \langle\phi|\phi\rangle$. Therefore, the diamond norm of ITE is 
\begin{equation}\label{eq:diamond_ITE}
    \da{\mathcal{T}} = \max_{\ket{\phi}} a_\phi = e^{-2\beta\lambda_{0}},    
\end{equation}
where $\lambda_{0}$ is the smallest eigenvalue of $H$.

If $H$ is positive definite (PD), the eigenvalues $\lambda_i > 0$, $\da{\mathcal{T}} = e^{-2\beta\lambda_{0}}$. As $\beta$ increases, the optimal sampling cost $\gamma_\text{opt}$ gets smaller; if $H$ is positive semi-definite (PSD), $\lambda_{0} = 0$, $\da{\mathcal{T}} = 1,$ i.e. the optimal cost $\gamma_\text{opt}$ will not change with $\beta$; if $H$ is indefinite, the smallest eigenvalue $\lambda_{0}$ is negative, then the cost will increase as $\beta$ increases. %

Based on our analysis, introducing a scaled identity operator into the Hamiltonian appears to improve the sampling cost. However, since imaginary-time evolution is not trace-preserving, the signal strength $\tr(\mathcal{T}(\rho))$ may not be maintained. 
In the following, we study the balance between sampling cost and signal strength.

\begin{prop}\label{prop:bound_trace}
Given a fixed Hamiltonian $H$ and any given inverse temperature $\beta$, the trace $\tr(\mathcal{T}(\rho))$ of the output state from ITE is upper bounded by the diamond norm $\da{\mathcal{T}}$, and lower bounded by the rescaled overlap $\tr(\rho_{gs}\rho)$ between the input state $\rho$ and ground state $\rho_{gs}$ of $H$ (with ground state energy $\lambda_0$), i.e.
\begin{equation}\label{eq:bound_trace}
    \da{\mathcal{T}} \ge \trace(\mathcal{T}(\rho)) \ge e^{-2\beta\lambda_0}\trace(\rho_{gs}\rho).
\end{equation}
\end{prop}

\begin{proof}
    The first inequality is straightforward from the definition of the diamond norm. 
    The second inequality is also straightforward:
    \begin{align}
    \trace(\mathcal{T}(\rho)) &=  \trace(e^{-\beta H}\rho e^{-\beta H}) \nonumber\\
    &= \sum_i e^{-2\beta \lambda_i} \bra{i}U^\dag \rho U\ket{i} \label{eq:overlaps}
    \end{align}
    Every term in~\cref{eq:overlaps} is non-negative, and $e^{-2\beta\lambda_0}\trace(\rho_{gs}\rho)$ is only one term in the summation.
\end{proof}

Since we know that $\da{\mathcal{T}} = e^{-2\beta \lambda_0}$, \cref{eq:bound_trace} can be rewritten as 
\begin{equation}
    1 \ge e^{2\beta\lambda_0}\trace(\mathcal{T}(\rho)) \ge \trace(\rho_{gs}\rho).
\end{equation}
Similar to the diamond norm $\da{\mathcal{T}}$, the signal strength $\trace(\mathcal{T}(\rho))$ is also affected by the scalar $e^{-2\beta\lambda_0}$.
When $\lambda_0>0$, as $\beta$ increases, the signal also decreases. Any (unnormalized) expectation $\trace(\mathcal{T}(\rho)A)$ for a bounded observable $A$ will vanish as $\beta$ approaches infinity. Therefore the sampling cost for such values with a fixed precision $\epsilon$ also approaches zero. 
Conversely, when $\lambda_0 <0$, we get a signal amplification, causing an increase in sampling cost. 
We can avoid this issue when $\lambda_0 = 0$, i.e. $H$ is PSD.
Therefore, the optimal adjustment for adding a scaled identity term $\alpha I$, where $\alpha\in\mathbb{R}$, to the Hamiltonian $H$ is thus the one that shifts the ground-state energy to zero.

Determining the appropriate scaling factor $\alpha$ can be challenging in general, given that $\lambda_0$ is unknown. However, for a frustration-free Hamiltonian, a global rescaling can be obtained by adjusting each local term accordingly. When implementing the global ITE via Trotterization into local components, each local evolution step drives the state towards the ground state of that local term. Because the global Hamiltonian is frustration-free, this process ensures that the signal strength is maintained throughout the entire procedure.
This does not extend to frustrated Hamiltonians, thereby constraining the general applicability of our method. Nonetheless, it is important to note that frustration-free Hamiltonians can still encode computationally hard problems. A notable example is the Feynman–Kitaev mapping~\cite{kitaev2002classical}, which can map complex computational problems into such Hamiltonians.

\section{Estimating expectation values: proofs}\label{app:algcosts}
\begin{proof}[\bf Proof of \cref{lem:qpd-cost}]
    Recall \cref{eq:qpd-gen-map}
    \begin{align}
        \trace{\left[A \cT(\rho) \right]} =& \sum_{i=1}^D \sgn(q_i) \cdot |q_i| \cdot \trace{\left[A \cB_i(\rho)\right]}.
    \end{align}
    Let $I$ be a random variable that takes values in $\{1,..., D\}$ with probability  
    \begin{align}
        \text{Pr}[I=i] = \frac{q_i}{\gamma}.
    \end{align}  
    Let $L_I\in \{0, 1\}$ be the indicator variable that is 1 when the operation $\cB_I$ is successfully implemented, and let $Y_I$ denote the outcome of measuring the observable $A$.
    Note that
    \begin{align}
        \text{Pr}[L_i = 1] = 1, & \text{ if } \mathcal{B}_{i} \text{ is trace preserving,}\\
        \text{Pr}[L_i = 1] =\trace{\left[\cB_i(\rho)\right]} \leq 1, & \text{ if } \mathcal{B}_{i} \text{ is non-trace preserving.}
    \end{align}
    Now, we define the random variables 
    \begin{align}\label{eq:random-variables}
        F = \gamma \sgn(q_I) L_I Y_I \qquad G = \gamma \sgn(q_I) L_I .
    \end{align} 
    Computing their expectations below as
    \begin{align}
        \ex[F] &= \gamma \sum^{D}_{i=1}\sum_{L_i}\ex_{Y_i} \sgn(q_i) \text{Pr}[I = i] \text{Pr}[L_i = 1]Y_i \\
        &= \gamma \sum^{D}_{i=1} \frac{|q_i|}{\gamma} \sgn(q_i) \trace[A\cB_i(\rho)]\\
        &=  \trace[A\cT(\rho)] \label{eq:expct-obs}
    \end{align}
    and
    \begin{align}
        \ex[G] &= \gamma \sum^{D}_{i=1}\sum_{L_i} \sgn(q_i) \text{Pr}[I = i] \text{Pr}[L_i = 1]  \\
        &= \gamma \sum^{D}_{i=1} \frac{|q_i|}{\gamma} \sgn(q_i) \trace[\cB_i(\rho)]\\
        &=  \trace[\cT(\rho)] \label{eq:expct-norm},
    \end{align}
    we see that the rescaled expectation value is given by 
    \begin{align}
        \langle A\rangle_{\cT} = \frac{\ex[F]}{\ex[G]}
    \end{align}
    proving the correctness of \cref{alg:est-rescaled-exp}.
    
    We use standard concentration inequalities to bound the number of measurements needed to estimate $\langle A\rangle_{\cT}$ within a given precision, We note that both $F$ and $G$ lie between $\pm \gamma$. For $N$ independent trials (each trial consists of sampling an index $i$, running the corresponding quantum circuit and accepting/rejecting, and measuring $A$), define $\bar{F} = \frac{1}{N}(F_1+F_2+\dots+F_N)$ and $\bar{G} = \frac{1}{N}(G_1+G_2+\dots+G_N)$. Then, Hoeffding's inequality implies
    \begin{align}\label{eq:concentration-bound-1-app}
        \Pr\left[|\bar{F} - \ex[F]| \ge \epsilon \right] &\le \exp \left(-\tfrac{N\epsilon^2}{2\gamma^2} \right) \\
        \label{eq:concentration-bound-2-app}
        \Pr\left[|\bar{G} - \ex[G]| \ge \epsilon \right] &\le \exp \left(-\tfrac{N\epsilon^2}{2\gamma^2} \right).
    \end{align}
    Choosing
    \begin{align}\label{eq:num-samples-single}
        N = \frac{2\gamma^2}{\epsilon^2}\log(1/\delta)
    \end{align}
    ensures that the RHS of \cref{eq:concentration-bound-1-app,eq:concentration-bound-2-app} are at most $\delta$. Then, we have with probability at least $(1-\delta)$ 
    \begin{align}
        \left|\frac{\bar{F}}{\bar{G}} - \frac{\mathbb{E}[F]}{\mathbb{E}[G]}\right| &= \left| \frac{\mathbb{E}[G](\bar{F} - \mathbb{E}[F]) + \mathbb{E}[F](\mathbb{E}[G] - \bar{G})}{\bar{G}\mathbb{E}[G]} \right|\\
        &\leq \left| \frac{\mathbb{E}[G](\bar{F} - \mathbb{E}[F])}{\bar{G}\mathbb{E}[G]} \right| + \left|\frac{\mathbb{E}[F](\mathbb{E}[G] - \bar{G})}{\bar{G}\mathbb{E}[G]} \right|\\
        &= \frac{1}{|\bar{G}|} \left(\epsilon + \epsilon\left|\frac{\mathbb{E}[F]}{\mathbb{E}[G]}\right|\right)\\
        &= \frac{2\epsilon}{|\bar{G}|} \label{eq:concentration-bound-ratio-1-app}
    \end{align}
    where we have used \cref{eq:expct-obs,eq:expct-norm} and that $\tfrac{\mathbb{E}[F]}{\mathbb{E}[G]}\le 1$. Moreover
    \begin{align}
        \frac{1}{|\bar{G}|} \le \frac{1}{|\mathbb{E}[G]-\epsilon|} \le \frac{1}{|\mathbb{E}[G]|}\left(1+\frac{\epsilon}{|\mathbb{E}[G]|}\right). \label{eq:denom-G-app}
    \end{align}
    Using \cref{eq:concentration-bound-ratio-1-app,eq:denom-G-app} and recalling \cref{eq:expct-obs} we see that
    \begin{align}\label{eq:concentration-bound-final-single-app}
        \left|\frac{\bar{F}}{\bar{G}} - \langle A \rangle_{\cT} \right| \le \frac{2\epsilon}{\left|\trace{\left[\cT(\rho)\right]}\right|}\left(1+\frac{\epsilon}{\left|\trace{\left[\cT(\rho)\right]}\right|}\right)
    \end{align}
    with probability at least $(1-\delta)$ for choice of $N$ as in \cref{eq:num-samples-single}.

\end{proof}

\begin{proof}[\bf Proof of \cref{lem:qpd-cost-repeated}]
    The proof essentially follows that of \cref{lem:qpd-cost} with more involved notation.
    First note that for a $\mathcal{E}$ given by $R$ repeated applications of $\cT$ we have
    \begin{align}
        \trace{\left[A \mathcal{E}(\rho) \right]} =& \sum_{i_1=1}^D\sum_{i_2=1}^D\cdots \sum_{i_R=1}^D \sgn(q_{i_1}q_{i_2}\cdots q_{i_R}) |q_{i_1}q_{i_2}\cdots q_{i_R}| \cdot \trace{\left[A \cB_{i_R}\left(\cdots\cB_{i_2}\left(\cB_{i_1}(\rho)\right)\right)\right]}.
    \end{align}
    We define $I_{i_r}$ be a random variable that takes values in $\{1,..., D\}$ with probability  
    \begin{align}
        \text{Pr}[I_{i_r}=i] = \frac{q_i}{\gamma}.
    \end{align}  
    Similar to the proof of \cref{lem:qpd-cost}, for each $i_r$ we define $L_{i_r}\in \{0, 1\}$ to be the indicator variable that is 1 when the operation $\cB_{i_r}$ is successfully implemented. We have $L_{i_r}=1$ if $B_{i_r}$ is trace-preserving and 
    \begin{align}
        \text{Pr}[L_{i_r} = 1] =\frac{\trace{\left[\cB_{i_r}(\cB_{i_{r-1}}(\cdots\cB_{i_1}(\rho)))\right]}}{\trace{\left[\cB_{i_{r-1}}(\cdots\cB_{i_1}(\rho))\right]}}, 
    \end{align}
    if $\cB_{i_r}$ is non trace-preserving. We let $Y$ be the random variable denoting measurement outcome of the observable $A$ and define 
    \begin{align}
        F &:= \gamma^R\sum_{i_1=1}^D\sum_{i_2=1}^D\cdots \sum_{i_R=1}^D \sgn(q_{i_1}q_{i_2}\cdots q_{i_R}) \cdot L_{i_R}\cdots L_{i_2}L_{i_1} Y \\
        G &:= \gamma^R\sum_{i_1=1}^D\sum_{i_2=1}^D\cdots \sum_{i_R=1}^D \sgn(q_{i_1}q_{i_2}\cdots q_{i_R}) \cdot L_{i_R}\cdots L_{i_2}L_{i_1}.
    \end{align}
    The expectation values of $F$ and $G$ are given by
    \begin{align}
        \ex[F] &= \sum_{i_1=1}^D\sum_{i_2=1}^D\cdots \sum_{i_R=1}^D \sgn(q_{i_1}q_{i_2}\cdots q_{i_R}) |q_{i_1}q_{i_2}\cdots q_{i_R}| \cdot \trace{\left[\cB_{i_R}\left(\cdots\cB_{i_2}\left(\cB_{i_1}(\rho)\right)\right)\right]} \nonumber \\
        &= \trace{\left[A\mathcal{E}(\rho) \right]},\label{ex:avg-unnorm-repeated} \\
        \ex[G] &= \sum_{i_1=1}^D\sum_{i_2=1}^D\cdots \sum_{i_R=1}^D \sgn(q_{i_1}q_{i_2}\cdots q_{i_R}) |q_{i_1}q_{i_2}\cdots q_{i_R}| \cdot \trace{\left[\cB_{i_R}\left(\cdots\cB_{i_2}\left(\cB_{i_1}(\rho)\right)\right)\right]} \nonumber \\
        &= \trace{\left[\mathcal{E}(\rho) \right]}\label{ex:norm-repeated}
    \end{align}
    and thus
    \begin{align}
        \langle A \rangle_\mathcal{E} = \frac{\ex[F]}{\ex[G]}
    \end{align}
    proving correctness of \cref{alg:est-repeated}. Now, both $F$ and $G$ lie between $\pm \gamma^R$ (recall that $\|A\|\le 1$). For $N$ independent trials (each trial consists of sampling a multiset of indices $\{i_1,i_2,\dots,i_R\}$, running the corresponding quantum circuit and accepting/rejecting, and finally measuring $A$), we define $\bar{F} = \frac{1}{N}(F_1+F_2+\dots+F_N)$ and $\bar{G} = \frac{1}{N}(G_1+G_2+\dots+G_N)$. Then Hoeffding's inequality implies
    \begin{align}\label{eq:concentration-bound-repeated}
        \Pr\left[|\bar{F} - \ex[F]| \ge \epsilon \right] &\le \exp \left(-\tfrac{N\epsilon^2}{2\gamma^{2R}} \right) \\
        \Pr\left[|\bar{G} - \ex[G]| \ge \epsilon \right] &\le \exp \left(-\tfrac{N\epsilon^2}{2\gamma^{2R}} \right).
    \end{align}
    Following the same analysis as in the proof of \cref{lem:qpd-cost} (c.f.\cref{eq:num-samples-single} to \cref{eq:concentration-bound-final-single-app}), we see that choosing 
    \begin{align}\label{eq:num-samples-repeated}
        N = \frac{2\gamma^{2R}}{\epsilon^2}\log(1/\delta)
    \end{align}
    ensures that
    \begin{align}\label{eq:concentration-bound-final-repeated}
        \left|\frac{\bar{F}}{\bar{G}} - \langle A \rangle_{\mathcal{E}} \right| \le \frac{2\epsilon}{\left|\trace{\left[\mathcal{E}(\rho)\right]}\right|}\left(1+ \frac{\epsilon}{\left|\trace{\left[\mathcal{E}(\rho)\right]}\right|}\right)
    \end{align}
    with probability at least $(1-\delta)$.
\end{proof}

\section{Trotterized ITE: error bounds}\label{app:ite-cost-proofs}
Here we determine the cost of estimating imaginary-time evolved expectation values to within a given precision. 

We define $P = e^{-\beta H}$ and denoted its $1^{\mathrm{st}}$ order Trotter approximation (see \cref{eq:productapproxHam}) as $T = \left(e^{-\frac{\beta}{r} H_1}e^{-\frac{\beta}{r}  H_2}...e^{-\frac{\beta}{r}  H_L}\right)^r$.
The corresponding maps are $\cT(\cdot)$ and $\mathcal{P}(\cdot)$ i.e., $\mathcal{P}(\cdot) = P(\cdot)P$ and $\cT(\cdot) = T(\cdot)T$.
From Theorem 6 in \cite{childs2021theory} it can be shown that
\begin{align}\label{eq:first-trot-error}
    || T - P || &= \mathcal{O}\left( \left(\sum^{L}_{l=1} || H_l || \beta \right)^2 e^{\sum^{L}_{l=1} ||H_l|| \beta} \right)\\
    &= \mathcal{O}\left(\frac{\beta^2 L^2}{r}  e^{\beta L/r} \right) \label{eq:trot=error-op}.
\end{align}
For an input state $\rho$, we can bound the error in the expectation value $\tr{\left[Ae^{-\beta H} \rho e^{-\beta H}\right]}$ for an observable $A$ as below
\begin{align}
    |\trace{[A\cT(\rho)]} - \trace{[A\mathcal{P}(\rho)]}| &= |\trace{[AT\rho T]} - \trace{[AP\rho P]}| \nonumber \\
    &= |\trace{[AT\rho T]} - \trace{[AT\rho P]} + \trace{[AT\rho P]} - \trace{[AP\rho P]}| \nonumber \\
    &= |\trace{[AT\rho (T- P)]} + \trace{[A(T-P)\rho P]}| \\
    &\leq |\trace{[AT\rho (T- P)]}| + |\trace{[A(T-P)\rho P]}|\\
    &\leq ||A|| \cdot ||T|| \cdot ||(T- P)|| + ||A|| \cdot ||(T-P)|| \cdot ||P||\\
    &\leq 2 ||e^{-\beta H}|| \cdot \mathcal{O}\left(\frac{\beta^2 L^2}{r} e^{\beta L/r} \right)\\
    &= \mathcal{O}\left(\frac{\beta^2 L^2 }{r} e^{\beta L/r} \right) \label{eq:err-obs-trotter}
\end{align}
where we have used \cref{eq:trot=error-op} in the second-last step, and that $H$ is positive semidefinite in the last step.
Similarly,
\begin{align}
    |\trace{[\cT(\rho)]} - \trace{[\mathcal{P}(\rho)]}| &= \mathcal{O}\left(\frac{\beta^2 L^2 }{r} e^{\beta L/r} \right)
\end{align} 

We choose $r = O\left(\tfrac{\beta^2 L^2}{\epsilon}\right)$ such that the error on the RHS above and thus also in \cref{eq:err-obs-trotter} are both at most $O(\epsilon)$. 

In our case, each term $H_i$ in the Hamiltonian is of the form $h\otimes I$ where $h$ is a $k$-local interaction, meaning that we have the conditions as in \cref{lem:qpd-cost-repeated} where we repeatedly apply the same local map to different subsets of qubits.

\section{Further implementation details}

For the demonstration on hardware shown in \cref{HeisExpt}, we also were able to run circuits for the fifth Trotter step.
However, due to retirement of the device (5-qubit IBM systems) in November 2023, only 20000 QPD samples were collected, where the circuit corresponding to each sample was run for 512 shots.
The energy estimate we computed from this data is $-0.275$, which should be compared with the ideal value of $-0.269$ and the simulated value of $-0.275\pm0.010$.

\end{document}